%% file: content_aps.tex
\documentclass[%
 reprint,
 amsmath,amssymb,
 aps,
 prx,
]{revtex4-2}

\usepackage{graphicx}
\usepackage[svgnames]{xcolor}
\usepackage{dcolumn}
\usepackage{bm}

\input{prefix_aps}

\begin{document}

\preprint{APS/123-QED}

\title{Acquisition of delocalized information via classical and quantum carriers}
\author{Julian Maisriml}
\author{Sebastian Horvat}
\altaffiliation[]{University of Vienna, Department of Philosophy, Universitätsstraße 7, 1010 Vienna, Austria}
\author{Borivoje Dakić}
\altaffiliation[]{Institute for Quantum Optics and Quantum Information (IQOQI), Austrian Academy of Sciences, Boltzmanngasse 3, A-1090 Vienna, Austria}
\affiliation{University of Vienna, Faculty of Physics, Vienna Center for Quantum Science and Technology, Boltzmanngasse 5, 1090 Vienna, Austria}




\date{\today}

\begin{abstract}
We investigate the information-theoretic power of spatial superposition by analyzing tasks in which information is locally encoded at multiple distant sites and must be acquired by a single information carrier, such as a particle. Within an operational framework, we systematically compare the statistical correlations that can be generated in such tasks using classical particles, quantum particles in spatial superposition, and \changes{more general} “second-order interference” resources. We bound classical strategies via convex polytopes and present a study of their symmetry, demonstrating that the vertices are inherently connected to $K$-juntas as defined in the classical theory of Boolean functions, while their facet inequalities are in one-to-one correspondence with oracle games. We then analyze the violation of the ``fingerprinting inequality'' achievable by the use of one quantum particle, and we study the dependence of this violation on the dimension $d$ of the particle's internal degree of freedom. \changes{In particular, we show that the case of $d=2$ can achieve a higher violation of the inequality than the previously investigated case of $d=1$. We also provide analytic and numerical evidence that this violation cannot be further increased for larger $d>2$.
Finally}, we find that both quantum and any other (generalized) second-order interference models exhibit the same asymptotic scaling in violating the fingerprinting inequality. Our results thereby further articulate quantum interference and spatial superposition as a resource for information processing.
\end{abstract}

\maketitle

\section{Introduction}
Quantum information theory, broadly construed, examines the capabilities and limitations of utilizing quantum-mechanical systems for encoding, transmitting, and decoding information. One feature that distinguishes quantum from classical mechanics is quantum entanglement, which accordingly lies at the basis of many protocols that outperform the best-known classical counterparts in computation, communication, and cryptography. Another distinguishing feature of quantum theory—which has received relatively less attention in the quantum information community—is that quantum particles can be prepared in a superposition of multiple locations. This alludes to a potential enhancement in the processing of information that is dispersed throughout space. Following this cue, recent theoretical and experimental studies have shown that a single quantum particle can outperform a plurality of classical particles in tasks requiring the transmission of “delocalized” information. 

In their work, Del Santo and Daki\'c \cite{delsanto_twoway_2018} showed that quantum superposition allows two distant parties to engage in effective two-way communication using only a single particle—something classically impossible under relativistic constraints. This result framed interference as a resource for communication. Building on this, Horvat and Dakić \cite{horvat_quantum_2019,horvat_quantum_2021,horvat_interference_2021} reformulated multislit experiments as information-theoretic parity games, connecting the order of interference to computational advantage. 
Zhang and coworkers \cite{zhang_building_2022,chen_information_2024} introduced a framework for multiple access channels (MACs) realized with single particles, distinguishing quantum from classical encoding strategies. Quantum MACs were shown to violate generalized fingerprinting inequalities, marking the boundary of classical resources in interferometric-like setups. These insights and developments also connect to broader foundational studies of quantum communication through devices with an indefinite input-output direction \cite{delsanto_coherence_2020,liu_quantum_2023} and the so-called second-quantized Shannon theory \cite{chiribella_quantum_2019}. 

In parallel, the use of spatial superposition in quantum cryptography has inspired new protocol designs. Massa \textit{et al.} proposed a communication scheme in which a single quantum particle is exchanged only once between two parties, enabling secure and deterministic key establishment \cite{massa_experimental_2019}. Building on these initial ideas, semi-device-independent quantum key distribution protocols have been designed \cite{massa_experimental_2022,silva_coherencewitnessing_2023}. Another practical aspect enabled by quantum spatial superposition is the possibility of detecting quantum coherence in interferometric experiments without path reinterference \cite{delsanto_coherence_2020,kun_direct_2025} or its detection with a single setting in triangular scenarios \cite{bibak_quantum_2024}.

In this paper, we present a systematic analysis of spatial superposition as an information-theoretic resource, framed through the lens of information acquisition tasks involving data that is dispersed and delocalized across space. We introduce an operational framework that enables a clear distinction between classical, quantum, and higher-order interference theories. In the spirit of classical polytope constructions in the context of Bell’s theorem and nonlocality, we undertake a structured study in which classical strategies form a convex polytope; its facets correspond to inequalities [such as the fingerprinting inequality \eqref{fingerprint_ineq}] and oracle games that serve to demarcate classical from quantum resources. The framework also serves to compare the quantum case to generalized “second-order interference” models \cite{sorkin_quantum_1994} (see also Ref. \cite{lee_higherorder_2017}). Our results support the view that interference constitutes a fundamental resource for information processing in physical theories. 

The paper is structured as follows. In Sec. II, we introduce the setup of the tasks described above and study classical resources (information carriers or particles) and probabilistic correlations they generate as embedded in the convex polytope structure. We present a study on the symmetry of these polytopes, and we show that the vertices are directly linked to $K$-juntas \cite{odonnell_analysis_2014}, as defined in classical Boolean function theory, while the facet inequalities are in direct correspondence with oracle games. In Sec. III, we proceed with a thorough analysis of quantum-mechanical protocols, provide an analysis of the fingerprinting inequality \eqref{fingerprint_ineq}, and analyze the dependence of its violation on the dimensionality $d$ of the internal degree of freedom of the used particle. \changes{We thereby show that the case of $d=2$ can achieve a higher violation than the previously investigated case of $d=1$, and we provide analytic and numerical evidence that this violation cannot be further increased for larger $d>2$.} Finally, we compare the quantum violations for these latter setups with general ``second-order interference'' protocols, showing that they achieve asymptotically the same violation. Section IV concludes and indicates future directions that we deem worthy of pursuit.

\section{Operational setting and classical behaviors}
As stated in the introduction, we will be interested in the transmission of information dispersed throughout space. In particular, our setup will  consist of (1) $N$ bits $\mathbf{x}\equiv \left(x_1,...,x_N\right)$ that are encoded at $N$ distinct locations, (2) a source that emits particles whose purpose is to collect the latter information by interacting with the objects wherein the information is encoded, and (3) a device that produces a bit $a$ upon interaction with the particles, and whose aim is the decoding of the sought information. The experiments that we will analyze will be subdividable into three temporal stages:\\
$t_1$: the input bits $\mathbf{x}$ are fixed at their pertaining locations, and the source emits a number of particles, whose state does not depend on $\mathbf{x}$;\\
$t_2$: the particles interact with the objects (devices) wherein the input bits are encoded, in such a way as not to increase the number of particles;\\
$t_3$: the particles interact with the decoding device and an output bit $a$ is produced.

Figures \ref{fig:sub1} and \ref{fig:sub2} illustrate examples of these experiments for the special cases of $N=4$ bits encoded in spatially separated devices, where the particle source emits either three classical particles or one quantum particle.

\subsection{Probabilistic behaviors and convex polytopes}
As our interest lies in the statistical correlation between the ``input'' bits $\mathbf{x}$ and the ``output'' bit $a$ that are generated in single runs of our experiments, we will be concerned with conditional probabilities $P(a|\mathbf{x)}$, i.e., the probabilities that the decoding device produces output $a$ conditioned on the inputs having been set to $\mathbf{x}$. We will henceforth refer to complete lists of such probability distributions $\left(P(a|\mathbf{x}) \right)_{\mathbf{x}}$ as {behaviors}. 

\begin{figure}
\centering
\begin{subfigure}{.5\textwidth}
  \centering
  \includegraphics[width=.9\linewidth, angle=0]{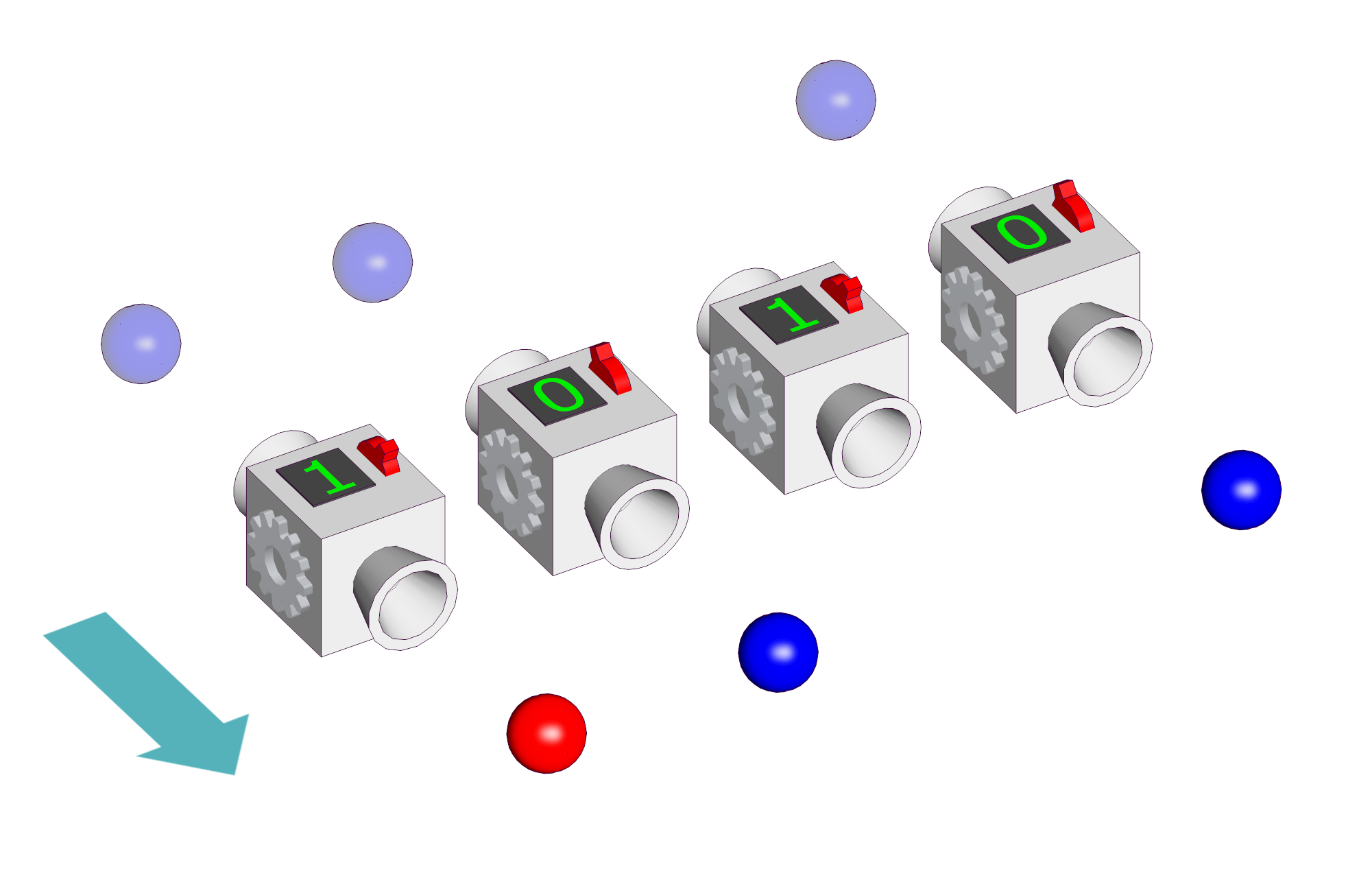}
  \caption{three classical particles\newline}
  \label{fig:sub1}
\end{subfigure}%
\hfill
\begin{subfigure}{.5\textwidth}
  \centering
  \includegraphics[width=.9\linewidth, angle=0]{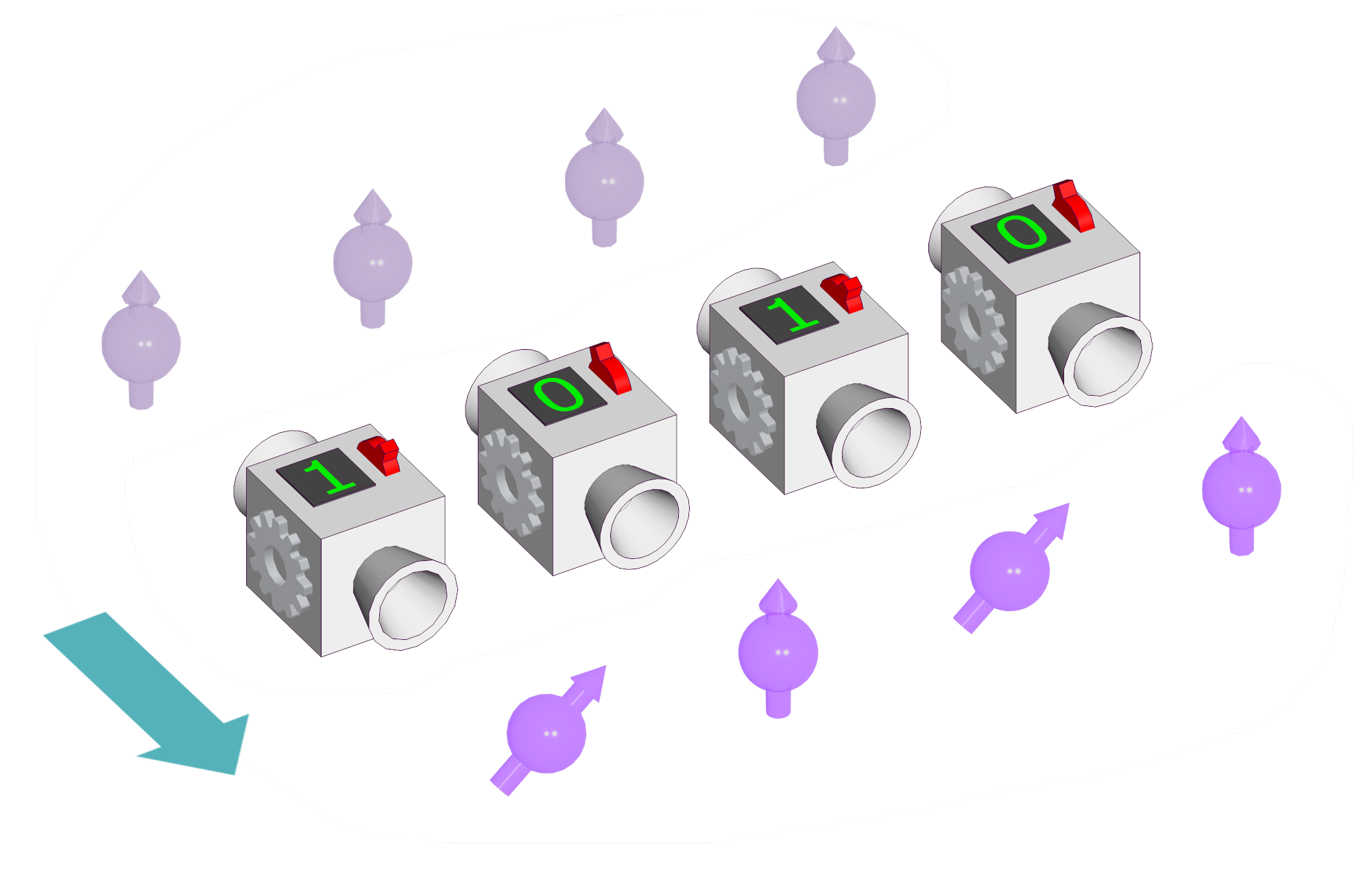}
  \caption{one quantum spin 1/2 particle in spatial superposition}
  \label{fig:sub2}
\end{subfigure}
\caption{Examples of our setup for $N=4$ bits featuring, respectively, three classical particles and one quantum particle.}
\label{fig:setup}
\end{figure}

Note that the set of all logically possible behaviors---namely those that satisfy the normalization constraint for probability measures---can be represented as a convex subset $\mathcal{L}_N$ of a $2^{N+1}$-dimensional real vector space: indeed, a canonical basis can be so chosen that the components of any vector $\Vec{P} \in \mathcal{L}_N$ relative to this basis are equal to the elements $P(a|\mathbf{x})$ of some behavior. On the other hand, we expect the sets of behaviors that can be generated with the use of $K<N$ classical particles to be smaller than the logical set; i.e., we expect that some behaviors can be generated by the use of at least $N$ classical particles, but not by $K<N$ classical particles. In accord with definitions introduced in previous works \cite{horvat_quantum_2021, zhang_building_2022}\footnote{\changes{Note that our definition of $\CNK$ coincides with the definition of $\mathcal{C}_{N, K}^{(sep)}$, as introduced in Ref.\cite{zhang_building_2022}. However, the polytope thereby defined also turns out to coincide with $\CNK$ of Ref. \cite{zhang_building_2022}: Indeed, Proposition 2 in Ref. \cite{zhang_building_2022} guarantees that this is so for cases, such as ours, in which output $a$ can only take values in $\{ 0,1\}$. This is why we chose to simplify the notation, and refer to our polytope as $\CNK$
}}, let us define the classical set of behaviors $\mathcal{C}_{N,K}$ as the largest subset of $\mathcal{L}_N$ whose elements $\Vec{P}$ all satisfy the following condition: there exists a set of non-negative and normalized weights $q_{j_1\cdots j_K}$ and a set of probability distributions $P(a|x_{j_1}\cdots x_{j_K})$, such that each canonical component $P(a|\mathbf{x})$ of $\Vec{P}$ satisfies

\begin{equation}
    P(a|\mathbf{x})=\sum_{j_1\neq\cdots\neq j_K=1}^N q_{j_1\cdots j_K} P(a|x_{j_1}\cdots x_{j_K}).
\end{equation}

Each element of $\mathcal{C}_{N,K}$ is producible by $K$ classical particles being sent with probability $q_{j_1\cdots j_K}$ through locations that host inputs $x_{j_1},...,x_{j_K}$; conversely, all behaviors that can be generated by $K$ particles lie in $\mathcal{C}_{N,K}$. It is simple to inspect that $\mathcal{C}_{N,K} \subset \mathcal{C}_{N,K+1} \subset \mathcal{L}_N$ for all $K<N$, and that all sets $\mathcal{C}_{N,K}$ are convex polytopes.

Furthermore, it holds that
\begin{equation}
    \mathcal{C}_{N,K} \subseteq \mathcal{J}_{N,K},
\end{equation} 
where $\mathcal{J}_{N,K}$ is the ``$K$th order interference set'' \cite{zhang_building_2022}---that is, the largest subset of $\mathcal{L}_N$ whose elements satisfy 
\begin{equation}\label{interf}
    \sum_{x_{j_1}\cdots x_{j_M}=0}^1 (-1)^{\sum_lx_{j_l}}P(0|\mathbf{x})=0,
\end{equation}
for all $M>K$ and all subsets $\left\{j_1,...,j_M \right\}\subseteq \left\{1,...,N \right\}$. $\mathcal{J}_{N,K}$ contains those behaviors that exhibit $K$th order interference, in a generalized ``semi-device-independent'' sense of interference \cite{horvat_interference_2021, zhang_building_2022}. Accordingly, it has been shown that the set of behaviors that can be generated in quantum versions of the above experiments with one quantum particle is fully contained in $\mathcal{J}_{N,2}$, which is a consequence of the order of interference of quantum theory being limited to $K=2$. On the other hand, $\mathcal{C}_{N,1} \subseteq \mathcal{J}_{N,1}$ is related to the fact that a classical particle cannot ``interfere with itself.'' We will have more to say on the relationship between quantum-mechanically producible behaviors and $\mathcal{J}_{N,2}$ in Sec. III. But before getting into all of this, let us lay out some observations that concern $\mathcal{C}_{N,K}$ and $\mathcal{J}_{N,K}$.


\subsection{$K$-juntas and symmetries of classical polytopes}
\textbf{Remark 1} (On some properties of $\CNK$). Let us point out some previously unnoticed properties of $\CNK$. First, note that a behavior $P\in \mathcal{L}_N$ is a vertex of $\CNK$ if and only if its components satisfy $P(a|\mathbf{x})=\delta_{a,f(\mathbf{x})}$ for some $K$-junta \cite{odonnell_analysis_2014}, that is, for some Boolean function $f: \{0,1\}^N \to \{ 0,1\}$ that effectively depends only on $K$ variables:
\[ f(x_{j_k} =0) =  f(x_{j_k} =1),~~~~~ k = 1,..., N-K\]
for some subset $\left\{j_1,...,j_{N-K} \right\} \subseteq \left\{1,...,N \right\}$.
In other words, the set of vertices of $\CNK$ is in one-to-one correspondence with the set of $K$-juntas. Since the number of Boolean functions depending on all their $K$ variables is given by $ \sum_{r=0}^K (-1) ^{K-r} \binom{K}{r} 2^{2^r}$ [\cite{harrison_introduction_1965},p. 169] , a simple counting argument shows that the number of $K$-juntas and thus the number of vertices of $\CNK$ amounts to 
\begin{equation}
    \vert V(\CNK)\vert = \sum_{k=0}^K \binom{N}{k} \sum_{r=0}^k (-1) ^{k-r} \binom{k}{r} 2^{2^r}.
\end{equation}

The superexponential growth of the number of vertices ($\vert V\vert=\mathcal{O}\left(2^{2^K}\right)$) makes an explicit numerical calculation of properties of $\CNK$ hardly tractable even for low $N$. Nevertheless, there are a few observations that we will tentatively state about the geometry of the polytopes. 
We will first comment on the symmetries of $\CNK$, that is, the group of linear transformations that map the polytope to itself. Note that $\CNK$ exhibits an inversion symmetry around its center:
The transformation \begin{equation}
    I: P(a | \mathbf{x}) \mapsto 1- P( a | \mathbf{x}) = P(\neg a | \mathbf{x})
\end{equation}
is a representation of the cyclic group $Z_2$ that maps vertex associated to $K$-junta $f$ to vertex associated with $K$-junta $\neg f$. Moreover, we can identify further symmetry elements by introducing the hyperoctahedral group $B_N$ of order $N$ \cite{graczyk_hyperoctahedral_2005} [see also Ref. \cite{harrison_introduction_1965}, p. 148], which is the semidirect group product $S_N \ltimes Z_2^N$ of the symmetric group $S_N$ and $N$ copies of $Z_2$. The action of $b = \pi \ltimes(s_1,...,s_N)\in B_N$ on a Boolean function $f$ as well as its thereby induced representation $R_b$ is given by
\begin{align}
    &b f(x_1, x_2, x_3, \dots, x_N)\nonumber \\&~= ~f(s_1(x_{\pi(1)}), s_2(x_{\pi(2)}), s_3(x_{\pi(3)}), \dots, s_N(x_{\pi(N)})) \\
    &R_b P(a|x_1, x_2, x_3, \dots,x_N)\nonumber \\&~=~ P(a | s_1(x_{\pi(1)}), s_2(x_{\pi(2)}), s_3(x_{\pi(3)}), \dots, s_N(x_{\pi(N)}))~
\end{align}
with $s_i \in Z_2$ either being identity $s_i(x) = x$ or inversion $s_i(x) = \neg x$.
Since for $b \in B_N$ and $f$ being a $K$-junta, $b f$ is also a $K$-junta, the polytope $\CNK$ is symmetric under the representation $R_b$ as vertices are mapped to vertices in a bijective manner. This symmetry points to the following construction of $\CNK$ in terms of the convex hull of $\binom{N}{K}$ hypercubes embedded into the whole space, formalized in the following lemma: 
\begin{lemma}\begin{align*}
    \CNK = \text{conv} \left\{ R_b ~\left(\mathcal{C}_{K,K} \oplus 0^{N-K}\right) \vert b = \sigma \ltimes {1\!\!1}^N, \sigma \in S\right\},
\end{align*}\label{cnk_convex_hull}
where $\sigma \in S\subset S_N$ if there exists $\left\{g_1,...,g_K \right\}\subset\left\{1,...,N \right\}$, such that $\sigma(i)=g_i$, for all $i=1,...,K$.\end{lemma}
\begin{proof}
Let $f(\mathbf{x})$ be an arbitrary $K$-junta that depends on variables $\left\{g_1,...,g_K \right\}\subset\left\{1,...,N \right\}$ and let $\sigma \in S_N$ be any permutation that satisfies $\sigma(i)=g_i$. Define $f'(x_1,...,x_N)=f(x_{\sigma^{-1}(1)},...,x_{\sigma^{-1}(N)})$, which is by construction a $K$-junta that depends on its first $K$ variables. It holds that $f(\mathbf{x})=(\sigma \ltimes \identity^N) f'(\mathbf{x})$. The statement in Lemma 1 now follows by noting that the set of all vertices of $\CNK$ corresponds to the set of all $K$-juntas, and that the set of all vertices of the $\mathcal{C}_{K,K}\oplus 0^{N-K}$ hypercube corresponds to the set of $K$-juntas that depend on their first $K$ variables.
\end{proof}
An example of the above construction is illustrated in Figure \ref{fig:octahedron_sqares}. As is known from previous investigations, $\mathcal{C}_{N,1}$ is a hyperoctahedron of dimension $N+1$ \cite{horvat_quantum_2019}. 
Its full symmetry group is $B_{N+1}$ and is thus larger than $B_N$, which was previously discussed. Consequently, the inversion symmetry and $B_N$ only constitute subgroups of the polytope's complete symmetry, leaving an exhaustive description of the symmetry group of $\CNK$ for now elusive.

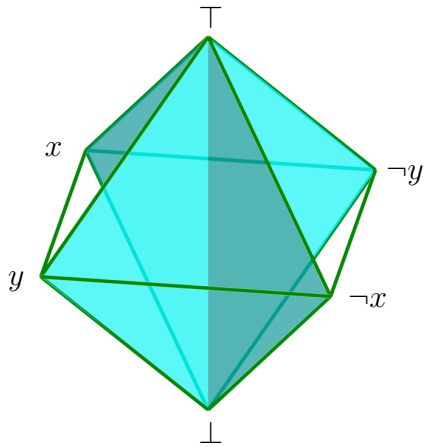
\begin{figure}
    \centering
    \begin{tikzpicture}
        \node[anchor=south west, inner sep=0] (image) at (0,0) {\includegraphics[width=0.6\linewidth]{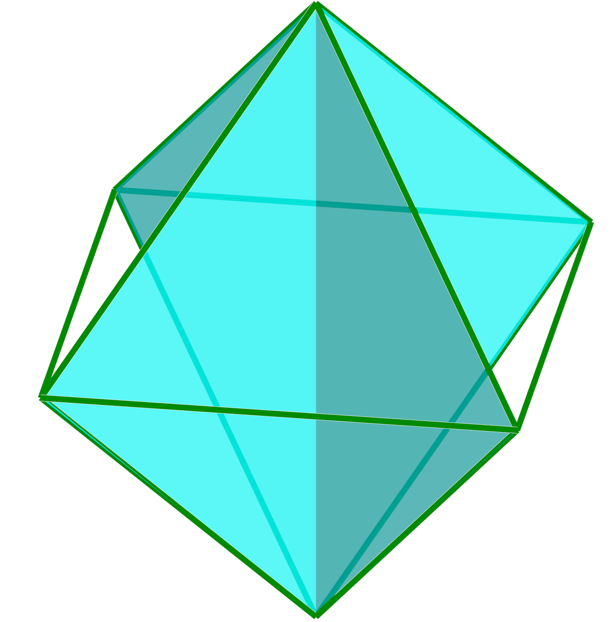}};
        \begin{scope}[x={(image.south east)},y={(image.north west)}]
            \node[black, font=\bfseries] at (0.52,-0.05) {$\bot$}; 
            \node[black, font=\bfseries] at (0.52,+1.05) {$\top$};
            \node[black, font=\bfseries] at (0.0,+0.35) {$y$};
            \node[black, font=\bfseries] at (0.1,+0.70) {$x$};
            \node[black, font=\bfseries] at (0.94,+0.30) {$\neg x$};
            \node[black, font=\bfseries] at (1.04,+0.63) {$\neg y$};
            the text here
        \end{scope}
    \end{tikzpicture}
    \caption{An octahedron ($\mathcal{C}_{2,1}$) as the convex hull of two squares ($\mathcal{C}_{1,1}$). The vertices correspond to Boolean 1-juntas in two variables.}
    \label{fig:octahedron_sqares}
\end{figure}

Furthermore, we want to mention that the search for the seven-dimensional polytope $\mathcal{C}_{3,2}$ with $f$-vector \footnote{An $f$-vector of a polytope indicates how many faces of given dimension it has, making it a useful characteristic for identifying polytopes} $(38, 408, 1608, 2764, 2208, 776, 96)$ in polytope databases \cite{klitzing_polytopes_2025,communityproject_polytope_2025} has not yielded any result, indicating that it has not been cataloged before. In case it has not been named, we suggest the name $(N,K)$-juntatope for $\CNK$. 



\bigskip

\textbf{Remark 2} (On some properties of $\mathcal{J}_{N,K}$): Now we will proceed with the interference polytopes $\mathcal{J}_{N,K}$ and show that their characterizing hyperplane equalities \eqref{interf} can be recast in a more compact form. Let $[k]$ denote the binary representation of the integer $k$, $[ k ]_j$ the $j$th bit thereof, and $\hamwgt{k}~$ the number of ones in $k$'s binary representation, also known as the ``Hamming weight.'' Consider a behavior $\Vec{P} \in \mathcal{L}_N$ and define a projection thereof as $\Vec{P}' = \left(P(1|~ [0]~),...,P(1|~ [2^N -1]~)\right)$. It can easily be inspected that $\Vec{P} \in \mathcal{J}_{N,K}$ if and only if the projection $\Vec{P}'$ satisfies
\begin{equation}\label{interf2}
    (H\Vec{P}')_j = 0, ~~~\text{for}~ \hamwgt{j} > K,
\end{equation}
where $H$ is the $2^N \times 2^N$ Hadamard matrix [\cite{horadam_hadamard_2007}, p. 11] with
\begin{equation}
    H_{i,j} = 2^{-N/2} (-1)^{\sum_{k=0}^{N-1} [i]_k [j]_k}.
\end{equation}

Furthermore, if $\Vec{P}$ is deterministic---that is, if its components satisfy $P(a|\mathbf{x})=\delta_{a,f(\mathbf{x})}$ for some Boolean function $f$---condition \eqref{interf2} is equivalent to $f$ having a bounded Boolean Fourier degree as defined in \cite{odonnell_analysis_2014}: $\deg f\leq K$. Since the Fourier degree of a Boolean $K$-junta $f$ is bounded as $\deg f \leq K$, this implies the aforementioned relation $\mathcal{C}_{N, K} \subseteq \mathcal{J}_{N, K} $. The converse is however not true, as there are Boolean functions that have Fourier degree bounded by $K$, but that are not $K$-juntas. Nevertheless, a theorem about the Fourier degree of Boolean functions [\cite{odonnell_analysis_2014}, p. 70] implies that

\begin{equation}
    \left(\mathcal{J}_{N, K} \cap \mathcal{D}_N\right) \subset \mathcal{C}_{N, K2^{K-1}},
\end{equation}
where $\mathcal{D}_N$ is the set of all deterministic behaviors in $\mathcal{L}_N$. 
It should accordingly be noted that there are in general behaviors in $\mathcal{J}_{N, K}$ that cannot be decomposed as the convex sum of deterministic behaviors in $\mathcal{J}_{N, K}$ (i.e. $\text{conv}\left(\mathcal{J}_{N, K} \cap \mathcal{D}_N\right) \neq\mathcal{J}_{N, K}$) or equivalently, that some vertices of $\mathcal{J}_{N, K}$ are not deterministic. One such non-deterministic vertex is explicitly constructed later in the proof of Theorem \ref{theorem:2nd_order_viol}.
\subsection{Facet inequalities and oracle games}
\textbf{Remark 3} (Connection between hyperplanes and oracle games): Finally, let us note something that will come of use later, and that concerns behaviors that lie in the intersection of $\mathcal{L}_N$ and any arbitrary hyperplane, including those that correspond to facets of $\CNK$, or those that characterize $\mathcal{J}_{N,K}$. The condition that any vector needs to satisfy in order to lie within a general hyperplane can be generally cast as

\begin{equation}
    \sum_{ \mathbf{x}} \left(b^{(1)}_{\mathbf{x}} P(1 | \mathbf{x}) +  b^{(0)}_{\mathbf{x}} P(0 | \mathbf{x})\right) = B,
\end{equation}
where $b^{(i)}_{\mathbf{x}},B$ are arbitrary real coefficients. Using the normalization conditions for behaviors, it can be shown that a behavior $\Vec{P} \in \mathcal{L}_N$ is an element of the above hyperplane if and only if the following holds:
\begin{equation}
    \sum_{ v(\mathbf{x})=1} q^{(1)}_{\mathbf{x}} P(1 | \mathbf{x}) +  \sum_{ v(\mathbf{x})=0} q^{(0)}_{\mathbf{x}} P(0 | \mathbf{x}) = C,
\end{equation}
where $v:\left\{0,1\right\}^{\times N}\rightarrow \left\{0,1\right\}$ is a function that satisfies $v(\mathbf{x})=0 \leftrightarrow b^{(0)}_{\mathbf{x}}>b^{(1)}_{\mathbf{x}}$, and where the coefficients $q^{(a)}_{\mathbf{x}}$ and $C$ are defined as 
\begin{align}
    q^{(a)}_{\mathbf{x}}=\frac{b^{(a)}_{\mathbf{x}}-b^{(a \oplus 1)}_{\mathbf{x}}}{\sum_{ v(\mathbf{x})=1}(b^{(1)}_{\mathbf{x}}-b^{(0)}_{\mathbf{x}})+\sum_{ v(\mathbf{x})=0}(b^{(0)}_{\mathbf{x}}-b^{(1)}_{\mathbf{x}})}, \\ C=\frac{B-  \sum_{\mathbf{x}} \min{(b^{(0)}_{\mathbf{x}} , b^{(1)}_{\mathbf{x}})}}{\sum_{ v(\mathbf{x})=1}(b^{(1)}_{\mathbf{x}}-b^{(0)}_{\mathbf{x}})+\sum_{ v(\mathbf{x})=0}(b^{(0)}_{\mathbf{x}}-b^{(1)}_{\mathbf{x}})}.
\end{align}

By construction, all the coefficients $q^{(a)}_{\mathbf{x}}$ that appear on the LHS are non-negative and sum up to one. The condition that a behavior $\Vec{P}$ needs to satisfy in order to be an element of a hyperplane can thus be interpreted as a constraint on the average probability that $\Vec{P}$ wins a game whose goal is to output $a$ for all those inputs $\mathbf{x}$ that satisfy $v(\mathbf{x})=a$, whereby the inputs are priorly distributed according to $\max (q^{(0)}_{\mathbf{x}},q^{(1)}_{\mathbf{x}})$. In other words, the goal of the game is the computation of function $v(\mathbf{x})$, for a given distribution of inputs $\mathbf{x}$. To each hyperplane, one can therefore associate its corresponding game and a constraint on its winning probability. Since any polytope can be fully characterized by facet inequalities, it follows that $\CNK$ and $\mathcal{J}_{N,K}$ can be characterized by bounds on the probabilities of winning the here introduced games. In the next section we will analyze one such game that corresponds to a facet of $\mathcal{C}_{N,N-1}$.
\section{Quantum behaviors}

Let us proceed with the quantum mechanical analysis of the experiments presented in Sec. II. We will specifically be concerned with setups in which the correlations between the input bits $\mathbf{x}$ and the output bit $a$ are established by one quantum particle with $d$ orthogonal ``internal'' states. The set of pure states associated to the particle will thus be given by $\mathcal{H}\equiv\mathbb{C}^N\otimes\mathbb{C}^d$, where $\mathbb{C}^N$ models the particle's spatial degree of freedom: We will henceforth denote with $\left\{\ket{1},...,\ket{N} \right\} \subset \mathbb{C}^N$ the basis that corresponds to ``interferometric-like'' semiclassical trajectories possibly traversed by the particle. Furthermore, the input bits $\mathbf{x}$ will be encoded in spatially separated devices that implement unitary transformations on the internal degrees of freedom of the said particle.\footnote{Our investigation is thus aimed at a smaller class of experiments than the one investigated in Ref. \cite{zhang_building_2022}, which studies general number preserving operations, unitary ones being only special cases thereof.} More precisely, each device--say, the $j$th one---is assumed to locally implement one of the two unitary transformations $U_j^{(0)}$,$U_j^{(1)}:\mathbb{C}^d\rightarrow \mathbb{C}^d$ depending on the value of $x_j$. The joint action of the $N$ unitaries is thus assumed to be representable in terms of a controlled gate $U_{\mathbf{x}}:\mathcal{H}\rightarrow \mathcal{H}$ that satisfies 
\begin{equation}
  U_{\mathbf{x}}=\sum_{j=1}^N\ket{j}\bra{j}\otimes U_j^{(x_j)},  
\end{equation}
for some list of unitaries $\left(U_j^{(x_j)}\right)$. Finally, the decoding device will be unrestricted and thus modeled by an arbitrary POVM $\left\{\Pi_0, \Pi_1 \right\}$ on the particle's state space.

The temporal stages of the experiments will accordingly be subdivided as follows:\\
$t_1$: the settings of the unitary devices are fixed in accord with input bits $\mathbf{x}$, and the source emits a particle in state $\rho$ that is independent of $\mathbf{x}$;\\
$t_2$: the particle interacts with the unitary devices, its state being thereby transformed into $\rho_{\mathbf{x}}\equiv U_{\mathbf{x}}\rho U_{\mathbf{x}}^{\dagger}$;\\
$t_3$: the particle interacts with the decoding device and output bit $a$ is produced according to distribution $P(a|\mathbf{x})=\Tr\left[\Pi_a\rho_{\mathbf{x}} \right]$.

Let us denote with $\Qd$ the set of all behaviors that can be generated in the above experiments by all quantum-mechanically allowed triples $\left(\rho, U_{\mathbf{x}}, \left\{\Pi_0,\Pi_1 \right\} \right)$, for fixed number of inputs $N$ and dimension of the particle's internal space $d$. Whereas several results have already been proven that partially characterize $\mathcal{Q}_{N,1}^1$ \cite{horvat_quantum_2019,zhang_building_2022}, the case of arbitrary $d$ will only be explored in what follows. We will do so by investigating---both numerically and analytically---the maximal amounts of quantum-mechanical violations of a family of facet inequalities that partially characterize the classical polytopes $\CNK$. Before getting into these specific inequalities, let us lay out some general simplifications that turn out to hold for the maximization over $\Qd$ of the violation of any hyperplane inequality.

Recall that the condition that a behavior $\Vec{P}\in \mathcal{L}_N$ needs to satisfy in order to be contained in an arbitrary hyperplane can be cast in the following game-like form:
\begin{equation}\label{hyperplane}
    \sum_{ v(\mathbf{x})=1} q^{(1)}_{\mathbf{x}} P(1 |\mathbf{x}) +  \sum_{ v(\mathbf{x})=0} q^{(0)}_{\mathbf{x}} P(0 |\mathbf{x}) = C,
\end{equation}
where $q^{(a)}_{\mathbf{x}}\geq 0$ and $\sum_{a,v(\mathbf{x})=a} q^{(a)}_{\mathbf{x}}=1$. It is simple to check that the quantity on the LHS generated by an arbitrary quantum strategy $S_N^{(d)}=\left(\rho, U_{\mathbf{x}}, \left\{\Pi_0,\Pi_1 \right\} \right)$---henceforth denoted by ``winning probability'' $P_W(S_N^{(d)})$---can be cast as 
\begin{equation}\label{pwinning}
    P_W(S_N^{(d)})=q_0\Tr \left[\Pi_0\sigma_0\right] + q_1\Tr \left[\Pi_1\sigma_1\right],
\end{equation}
where $q_a\equiv \frac{1}{\sum_{v(\mathbf{x})=a}q^{(a)}_{\mathbf{x}}}$ and $\sigma_a \equiv \frac{1}{q_a}\sum_{v(\mathbf{x})=a}q^{(a)}_{\mathbf{x}}\rho_{\mathbf{x}}$.

$P_W(S_N^{(d)})$ quantifies the probability of correctly discriminating quantum states $(\sigma_0, \sigma_1)$---priorly distributed according to $(q_0,q_1)$---with the use of POVM $\left\{ \Pi_0,\Pi_1\right\}$. In what follows, we will study the optimization of this winning probability over the quantum set $\Qd$. Even though we will explicitly be concerned with the maximization of the LHS of Eq. \eqref{hyperplane}, all of the following will also apply to the minimization thereof, as the latter would reduce to the maximization of the same quantity up to the exchange of $q_{\mathbf{x}}^{(0)}$ and $q_{\mathbf{x}}^{(1)}$.

First, note that the form of the RHS of \eqref{pwinning} already suggests that the optimization over POVMs can be carried out by a straightforward application of Helstrom's bound \cite{helstrom_quantum_1969}. Namely, for a fixed pair of input state and unitary $\Tilde{S}_N^{(d)}\equiv\left(\rho, U_{\mathbf{x}} \right)$, the maximal value of $P_W(S_N^{(d)})$ is given by

\begin{equation}
\label{pw_helstrom}
    P_W(\tilde{S}_N^{(d)})\equiv\max_{\left\{\Pi_0,\Pi_1 \right\}} P_W(S_N^{(d)}) = \frac{1}{2} + \frac{1}{2} || q_0\sigma_0 - q_1\sigma_1||_1,
\end{equation}
where $||\cdot||_1$ denotes the  trace matrix norm. The optimization problem is thereby reduced to the optimization of $P_W(\tilde{S}_N^{(d)})$ over input states $\rho$ and unitaries $U_{\mathbf{x}}$. The latter problem, however, is generally not easily tractable. Nevertheless, several analytic simplifications reduce its complexity, thereby easing the subsequent numerical optimization, as summarized in the following Lemma.

\begin{lemma}
\label{max_p_wlog}
    The maximum value of $P_W(\tilde{S}_N^{(d)})$ can be obtained for $\tilde{S}_N^{(d)}=\left(\ket{\psi}\bra{\psi}, U_{\mathbf{x}}\right)$ that satisfies
    \begin{itemize}
        \item $\ket{\psi}=\sum_i \sqrt{p_i} \ket{i}\ket{\chi}$, for some $\ket{\chi}\in \mathbb{C}^d$ and $p_i\geq 0$, $\sum_i p_i=1$,
        \item $U_{\mathbf{x}}=\sum_i \ket{i}\bra{i} \otimes \left(U_i\right)^{x_i}$, for some list of unitaries $\left(U_1,...,U_N \right)$.
    \end{itemize}
\end{lemma}
\begin{proof}
    See Appendix \ref{appa}.
\end{proof}

The quantum strategy that optimizes the departure from an arbitrary hyperplane can thus be constructed out of a separable pure state with $N$ real coefficients $(\sqrt{p}_1,...,\sqrt{p}_N)$, and a list of $N$ $d\times d$ unitaries $(U_1,\dots,U_N)$. Now, note that the optimal violation does not depend independently on the input state and on the unitaries, but only on their combination $\rho_{\mathbf{x}}=\ket{\psi_{\mathbf{x}}}\bra{\psi_{\mathbf{x}}}$, with

\begin{equation}
    \ket{\psi_{\mathbf{x}}}=\sum_i \sqrt{p_i} \ket{i} \ket{\chi_i^{(x_i)}},
\end{equation}
where $\ket{\chi_i^{(0)}}=\ket{\chi}$ for all $i$, and where vectors $\ket{\chi_i^{(1)}}$ are arbitrary. An optimal quantum strategy can thus be fully specified by coefficients $(\sqrt{p}_1,...,\sqrt{p}_N)$ and by $(N+1)$ $d$-dimensional pure states---amounting to a total of $\left[N(2d-1)-1\right]$ independent real parameters. Besides a significant reduction in the dimension of the space over which $P_W(S_N^{(d)})$ is to be optimized, the above result has another interesting implication. Since $(N+1)$ $d$-dimensional vectors span at most a $(N+1)$-dimensional subspace of $\mathbb{C}^d$, it follows that for any strategy $S$ that uses a $d$-dimensional internal space, there is a strategy $S_N^{(d)'}$ that uses a $(N+1)$-dimensional internal space, and such that $P_W(S_N^{(d)})=P_W(S_N^{(d)'})$. Since this whole discussion holds for game-like quantities induced by arbitrary hyperplanes \eqref{hyperplane}, we have thereby proved that a behavior $P \in \Qd$ is an element of a hyperplane only if there is a behavior $P' \in \mathcal{Q}_{N,1}^{N+1}$ that is also an element of that hyperplane. In other words, no hyperplane separates any one element of $\mathcal{Q}_{N,1}^{d>N+1}$ from all the elements of $\mathcal{Q}_{N,1}^{N+1}$. By the hyperplane separation theorem, this implies the following relation:

\begin{lemma}
     $\Qd \subseteq \mathcal{Q}_{N,1}^{N+1}$.
\end{lemma}

Intuitively, if a particle with $d>N+1$ is used, then there are necessarily some dimensions of the particle's internal degree of freedom that are not employed in the transmission of the $N$ input bits $\mathbf{x}$. We now close the general discussion and focus on a particular family of inequalities that partially characterize the classical polytopes.

\subsection{Fingerprinting inequality}

We will now focus on a special class of facet inequalities of $\mathcal{C}_{N,N-1}$, dubbed the fingerprinting inequalities, given by
\begin{equation}
    \label{fingerprint_ineq}
    \frac{1}{N+1} P (0 | \underbrace{0\cdots 0}_{N}) +  \frac{1}{N+1} \sum_{k=1}^N P (1 | 0\cdots
 \underset{k\text{th}}{\underset{\uparrow}{1}} \cdots 0 ) \hspace{-0.5pt}\leq \hspace{-0.5pt}\frac{N}{N+1},
\end{equation} and corresponding to computing the function\\
$v(\mathbf{x}) = \bigwedge_i^N x_i$.

The interest in these inequalities lies in them partially characterizing classical polytopes for arbitrary $N$: the inspection of their quantum violation indeed tells us something about how one quantum particle compares to an arbitrary number $N$ of classical particles. Whereas previous works \cite{horvat_quantum_2021, zhang_building_2022} have shown that for any $N$, there is a behavior in $\mathcal{Q}_{N,1}^{1}$ that violates the inequality, the case of arbitrary $d$ and arbitrary input state coefficients $p_i$ has not yet been investigated. Our goal is to inspect how these two factors affect the violation of the fingerprinting inequalities. In what follows, we will present findings from numerical searches and an analytic result that holds for a restricted class of quantum strategies.

As proven in the previous subsection, we can assume without loss of generality that the optimal quantum violation of the fingerprinting inequality is obtained by a pure separable state 
\begin{equation}\label{input3}
    \ket{\psi}=\sum_i \sqrt{p_i} \ket{i}\ket{\chi},
\end{equation}
and by unitaries that satisfy $U_i^{(x_i)}=\left(U_i\right)^{x_i}$. Let $\ket{\chi_i}\equiv U_i \ket{\chi}$, for each $i$. Denoting the LHS of Eq. \eqref{fingerprint_ineq} with $P_W^*$, the optimal violation, for fixed input state and unitaries $\Tilde{S}_N^{(d)}=\left(\ket{\psi}, U_{\mathbf{x}} \right)$, is given by 
\begin{equation}
\label{pw_helstrom2}
    P^*_W(\tilde{S}_N^{(d)})\equiv \max_{\left\{\Pi_0,\Pi_1 \right\}} P_W^*(S_N^{(d)}) = \frac{1}{2} + \frac{1}{2} || q_1\sigma_1 - q_0\sigma_0||_1,
\end{equation}
where $(q_0,q_1)=\left(\frac{1}{N+1},\frac{N}{N+1} \right)$, and
\begin{align}
(\sigma_0,\sigma_1) = \left(\ket{\psi}\bra{\psi}, \frac{1}{N}\sum_i \ket{\psi_i}\bra{\psi_i} \right), \label{sigma0sigma1}
\end{align}
with $\ket{\psi_i}=\ket{\psi}+\sqrt{p}_i\ket{i}\left(\ket{\chi_i}-\ket{\chi} \right)$.

Let us denote the maximal violation of the fingerprinting inequality with 
\begin{equation}
    \delta_N^{(d)} \equiv\max_{\Tilde{S}_N^{(d)}}P^*_W(\tilde{S}_N^{(d)}) - \frac{N}{N+1}, \label{delta_d_N}
\end{equation}
for each pair $(N,d)$. \changes{It should be noted that, in comparison to previous studies \cite{horvat_quantum_2021,zhang_building_2022}, the fingerprinting inequality \eqref{fingerprint_ineq} and its violation $\delta$ is here scaled down by a factor of $(N+1)$ due to the normalization constraints.} As we could not compute the general form of $\delta_N^{(d)}$ analytically, we 
addressed the case of symmetric unitaries, i.e. the class of quantum strategies whose encoding unitaries satisfy $U_1=.\cdots=U_N\equiv U$, thereby proving the following theorem:
\begin{theorem}\label{thm:max_qm_uniform}
    Assuming $U_i=U_j$, the maximal violation of the fingerprinting inequality \eqref{fingerprint_ineq} is given by
    \begin{equation}
    \delta_N^{(d)}= \frac{1}{(N+1) \left(N^2-3 N+1\right)} ~~ \text{for}~ N > 3
    \label{expression_delta_d_symmetric}
    \end{equation}
    and is achieved for $d=2$ and a symmetric input state with $p_i = 1/N$.
\end{theorem}
\begin{proof}
    See Appendix \ref{appb}.
\end{proof}

We furthermore tackled the general problem---the one that does not assume $U_i=U_j$---numerically, obtaining the results displayed in figure \ref{fig:plot_violations} \footnote{\changes{Here are some details on the numerical optimization. We parametrize a quantum protocol by (1) unit vectors $\small{\ket{\chi_i^{(1)}} \in {\mathbb{C}}}$---represented in hyperspherical coordinates - which correspond to the outputs of the encoding unitaries' actions on some fixed $\ket{\chi^{(0)}}$; and by (2) the weights $p_1,...,p_{N-1}$, which characterize the input state. We then numerically optimize the trace norm in Eq. \eqref{pw_helstrom2}, where expressions \eqref{sigma0sigma1} need to be evaluated in every step. The optimization was carried out using the SLSQP algorithm with multiple randomly sampled initial conditions}}. Besides showing how $\delta_N^{(d)}$ depends on $N$ and $d$, \changes{the graph also displays the dependence of the violation on whether the input state's superposition coefficients are uniformly distributed or not (i.e., on whether $p_i=\frac{1}{N}$ or $p_i\neq\frac{1}{N}$)}. It is interesting to note that protocols with asymmetric superposition coefficients ($p_i\neq \frac{1}{N}$) can yield a higher violation than the ones with uniform coefficients---despite the fact that the fingerprinting inequality is itself symmetric under permutations of the inputs. Note that this gap between the performance with symmetric and asymmetric input states seems to occur only in the case of $d=1$, and that it does so only under the further condition that the unitaries themselves are asymmetric, in the sense of $U_i \neq U_{j\neq i}$. This observation is in concord with the one reached in \cite{chen_information_2024}, where the authors investigated the channel capacities of a setup similar to ours \changes{for $d=1$}, and found that a protocol with asymmetric encodings---analogous to our asymmetric unitaries---requires an asymmetric input state if it is to optimize the sought channel capacity. 

\changes{Importantly, note that the graph shows that the cases $d>1$ allow for a higher violation than the case $d=1$, the one that was previously discussed in Refs. \cite{horvat_quantum_2021, zhang_building_2022}. However, our numerical data indicate that the maximal violation is already achieved for $d=2$, and that it cannot be further improved by higher $d>2$: Indeed, the highest violation (\textcolor{Green}{$\bullet$}, \textcolor{Green}{$\times$}, \textcolor{red}{$\times$}) in figure \ref{fig:plot_violations} coincides with Eq. \eqref{expression_delta_d_symmetric}. We therefore propose the following conjecture: 
}

\begin{figure}
    \centering
    \includegraphics[width=8.5cm,height=!]{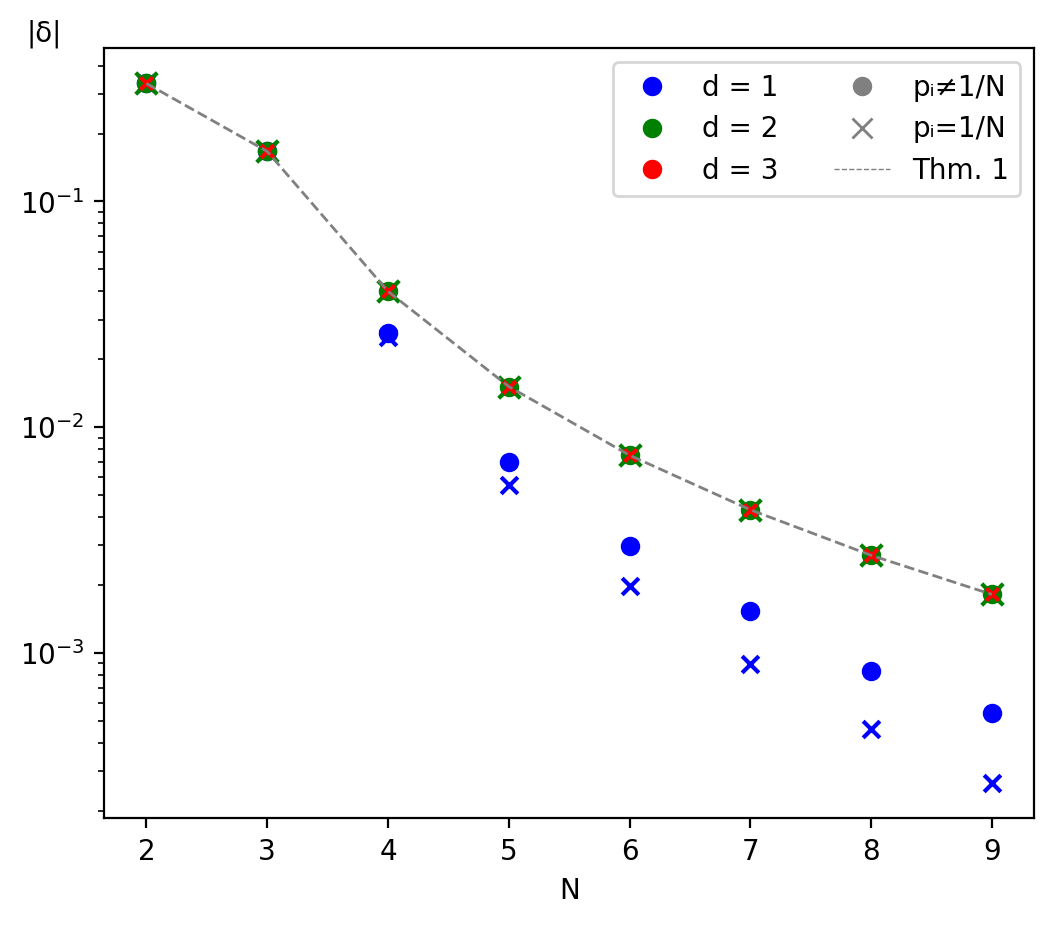}
    \caption{\changes{Numerically obtained }maximal violations $|\delta|$ of the fingerprinting inequality \eqref{fingerprint_ineq} for different quantum protocols. \changes{The blue cross (\textcolor{blue}{$\times$}) corresponds to the maximal violation calculated in previous works  \cite{horvat_quantum_2021,zhang_building_2022}. The improved violation represented by the dashed line corresponds to the result of our Theorem \ref{thm:max_qm_uniform}.}
    }
    \label{fig:plot_violations}
\end{figure}


\begin{conjecture}
\changes{The violation of the fingerprinting inequality \eqref{fingerprint_ineq} can be maximized with a quantum protocol with $d=2$ that obeys $U_i = U_j$. Therefore, the maximal violation is given by Theorem \ref{thm:max_qm_uniform}}.
\end{conjecture}
\subsection{General second-order interference behaviors}
In the previous subsections, we inspected some properties of $\Qd$ and more specifically, the amounts of quantum violations of the fingerprinting inequalities that correspond to particular facets of $\mathcal{C}_{N,N-1}$. In Sec. II, we had mentioned that all quantum behaviors that can be generated with the use of one particle cannot exhibit more than second-order interference, that is, $\Qd \subset \mathcal{J}_{N,2}$, for all $d$. This implies that an upper bound on quantum violations of the fingerprinting inequality is given by the maximal amount of violation thereof achievable by behaviors that lie in $\mathcal{J}_{N,2}$. Such a bound is presented in the following theorem.
\begin{theorem}
\label{theorem:2nd_order_viol}
The maximal violation $|\delta_{I2}|$ of the fingerprinting inequality \eqref{fingerprint_ineq} by behaviors that lie in $\mathcal{J}_{N,2}$ is given by 
\begin{equation}
    |\delta_{I2}| = \frac{2}{(N-2)\cdot (N-1)\cdot(N+1)} ~~ \text{for}~ N > 3\label{delta_I2_thm}
\end{equation}
and is achieved for 
\begin{equation}
   P (0 | \underbrace{0\cdots~0}_{N}) =\frac{2}{N^{2}-3N+2},~~~  P (1 | 0\cdots
 \underset{k\text{th}}{\underset{\uparrow}{1}} \cdots0 ) = 1. \label{P_I2:thm}
\end{equation}
\end{theorem}
\begin{proof}
    See Appendix \ref{appendix_I2}.
\end{proof}

Recall that Theorem \ref{thm:max_qm_uniform} provides a lower bound on quantum violations of the fingerprinting inequality that scales as $\mathcal{O}(N^{-3})$. Together with the above theorem, this implies that the unrestrictedly optimal quantum violation of the inequality also scales as $\mathcal{O}(N^{-3})$.

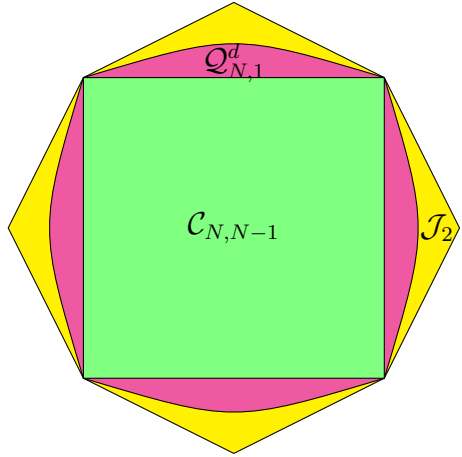
\begin{figure}
    \centering
    \begin{tikzpicture}[scale=2]
    \filldraw[fill=yellow]
    (-1.5,0) -- (-1,-1) -- (0,-1.5) -- (1,-1) -- (1.5,0) -- (1,1) -- (0,1.5) -- (-1,1) -- cycle;
    \filldraw[fill=magenta!80]
    (-1,1)  .. controls (-1.3,0) .. (-1,-1) .. controls (0,-1.3) .. (1,-1) .. controls (1.3,0) .. (1,1) .. controls (0,1.3) .. (-1,1) -- cycle;
    \filldraw[fill=green!50] (-1,-1) rectangle (1,1);
    \node[] at (0,0) {$\mathcal{C}_{N,N-1}$};
    \node[] at (0,1.1) {$\Qd$};
    \node[] at (1.35,0) {$\mathcal{J}_2$};
    \end{tikzpicture}
    \caption{Sketch of the sets $\mathcal{C}_{N,N-1}$, $\Qd$, and $\mathcal{J}_2$. See main text for details.}
    \label{fig:polysketch}
\end{figure}

Let us end our discussion by portraying the relationship between our sets of interest $\mathcal{C}_{N,N-1}$, $\Qd$, and $\mathcal{J}_2$; as in Fig. \ref{fig:polysketch}, inviting an analogy to the well-studied relationship between the classical, Bell-non-local, and the no-signaling polytopes that describe correlations producible in EPR-like experiments \cite{brunner_bell_2014}. Whereas we have shown that $\Qd$ and $\mathcal{J}_2$ roughly agree in their violation of the fingerprinting inequality, many questions concerning the relationship between the two polytopes remain unaddressed. \changes{For example, one could ask whether our result generalizes to other facet inequalities of $\CNK$, i.e., whether it is generally the case that quantum violations of classical inequalities saturate the second-order interference bound, and whether these violations approach to 0 in the limit of large $N$. While this certainly holds for those facets that are related to the fingerprinting inequality via some symmetry, we think we are not yet in a position to hypothesize about the general case. This and other such questions that concern the relationship between $\Qd$ and $\mathcal{J}_2$ are thus left for another occasion.}

\section{Conclusion and Outlook}

In this paper, we have furthered the study of the information-processing advantages provided by a spatially superposed quantum particle. We have focused on a particular class of experimental setups and shown various properties of the convex sets that describe correlations producible by classical and by quantum particles. In particular, we (1) pointed out mathematical properties of the classical polytopes, (2) analyzed how the dimensionality of the internal degrees of freedom of quantum particles and the (a)symmetricity of their initial states affect the violation of an inequality that cannot be breached with some specific number of classical particles, and finally, (3) deduced the scaling of the optimal quantum violation of the said inequality. \changes{In particular we showed that previous bounds on the violation $|\delta^{(1)}_N|$, which uses phase encoding only, can be surpassed by either considering an asymmetric superposition as input state or by utilizing higher-dimensional internal degrees of freedom in the encoding. However, it appears that the overall highest violation can already be obtained for $d=2$, symmetric input state ($p_i=1/N$), and symmetric encodings ($U_i = U_j$). For this case we provided an analytic expression of the violation $|\delta^{(2)}_N|$ in Theorem \ref{thm:max_qm_uniform}. }

Avenues still remaining to be investigated include to the very least the following ones. For once, the purely mathematical study of the geometry of the classical polytopes $\CNK$---and thereby of $K$-juntas - and of the quantum sets $\Qd$ is still far from complete. Furthermore, protocols with multiple quantum particles should be inspected, while keeping an eye on potential advantages obtainable from entanglement between particles. It would also be interesting to inspect some practical consequences of our results, along the lines of the discussion in  Ref.\cite{chen_information_2024}, which focuses on channel capacities. \changes{There, the authors investigate the communication rates of MACs that can be associated to our protocols, albeit only in the case of $d=1$. Since, as we have shown, the violation of the fingerprinting inequality can be increased by allowing higher-dimensional internal degrees of freedom, one might suspect that communication rates could also accordingly be increased for $d> 1$. If a higher bound were thereby to be found, it might also be possible to experimentally exploit this in a similar fashion as in Ref. \cite{chen_information_2024}, by, for instance, using a photon's polarization as the internal degree of freedom in the encoding scheme.} \changes{Moreover}, in regards to the second-order interference polytope, it is worth pondering on whether those of its behaviors that do not lie in the quantum polytope are somehow pathological, by e.g., conflicting with some basic physical or information-theoretic principle---analogous to the question concerning the potential pathology of non-quantum-mechanical no-signalling behaviors hypothetically arising in EPR-like experiments (see e.g., Ref. \cite{pawlowski_information_2009}).


\begin{acknowledgments}
This research was funded in
whole, or in part, by the Austrian Science Fund
(FWF) [10.55776/F71] and [10.55776/P36994]. For open access purposes,
the author(s) has applied a CC BY public copyright license to any author accepted manuscript version arising
from this submission.
\end{acknowledgments}
\appendix 
\section{Proof of Lemma \ref{max_p_wlog}}
\label{appa}

Since $P_W(\tilde{S}_N^{(d)})$ is linear in $\rho$, we can assume that the maximal amount of $P_W(\tilde{S}_N^{(d)})$ can be generated by some pure state $\ket{\psi}\bra{\psi}$, with general form
\begin{equation}\label{input}
    \ket{\psi}=\sum_{i,k}e^{i\phi_{ik}}\sqrt{p_{ik}}\ket{i}\ket{\chi'_k},
\end{equation}
where $\left\{\ket{\chi'_k}\right\}_k$ is an orthonormal basis on $\mathbb{C}^{d}$ and $\sum_{ik}p_{ik}=1$. Now we will show that this pure state can be assumed to be unentangled without loss of generality. Namely, take an arbitrary internal state $\ket{\chi}\in \mathbb{C}^{d}$, and define a list of unitaries $(G_1,...,G_N)$ that satisfy 
\begin{equation}
    G_i\ket{\chi}=\frac{1}{\sqrt{p_i}}\sum_k e^{i\phi_{ik}}\sqrt{p_{ik}}\ket{\chi'_k},
\end{equation}
where $p_i=\sum_k p_{ik}$. It is easy to see that the same states $\rho_{\mathbf{x}}$ defined by the possibly entangled input state $\ket{\psi}$ and by local unitaries $U_i^{(x_i)}$ is also generated by the separable state 
\begin{equation}\label{input2}
    \ket{\Tilde{\psi}}=\sum_i \sqrt{p_i} \ket{i}\ket{\chi},
\end{equation}
and by unitaries $\Tilde{U}_i^{(x_i)}\equiv U_i^{(x_i)}G_i$. The maximal value of $P_W(\tilde{S}_N^{(d)})$ can thus be assumed to be achievable by a separable pure state with positive coefficients as in Eq. \eqref{input2}, proving the first bullet-point of Lemma \ref{max_p_wlog}.

Another simplification comes after noticing that the unitaries can be recast as 
\begin{equation}
    U_i^{(x_i)}=U_i^{(0)}\left(U'_i \right)^{x_i},
\end{equation}
where $U'_i\equiv U_i^{(0)\dagger}U_i^{(1)}$. Indeed, note that for any strategy $S_N^{(d)}=\left(\rho, U_{\mathbf{x}}, \left\{\Pi_0,\Pi_1 \right\} \right)$, where $U_{\mathbf{x}}=\sum_i \ket{i}\bra{i}\otimes U_i^{(x_i)}$, there is another strategy $S_N^{(d)'}=\left(\rho, U'_{\mathbf{x}}, \left\{\Pi'_0,\Pi'_1 \right\} \right)$, such that $P_W(S_N^{(D)})=P_W(S_N^{(D)'})$, where $U_{\mathbf{x}}=\sum_i \ket{i}\bra{i} \otimes \left(U_i'\right)^{x_i}$ and $\Pi'_a=T^{\dagger}\Pi_a T$, with $T\equiv \sum_i \ket{i}\bra{i}\otimes U_i^{(0)}$. We can thus assume without loss of generality that each of the unitary devices implements the identity transformation when their input is set to 0, proving the second bullet point of Lemma \ref{max_p_wlog}.

\section{Proof of Theorem \ref{thm:max_qm_uniform}}
\label{appb}

The calculation of the optimal violation requires the optimization of the trace norm of $M\equiv q_1\sigma_1 - q_0\sigma_0$. Let us for compactness introduce the notation $\ket{\chi'}\equiv U\ket{\chi}$. A few lines of algebra show that $M$ takes the following form:

\begin{equation}
    M = N \sqrt{D_p} ~M_s \sqrt{D_p},
\end{equation}
where operator $D_p$ is the diagonalized $D_p=\sum_i p_i \ket{i}\bra{i} \otimes \identity_d$ that mediates the initial state's coefficients. $M_s$ denotes $M$ in the case of a symmetric input state ($p_i = 1/N$) and is given via Eq. \eqref{sigma0sigma1} by
\begin{equation}
    M_s=\frac{1}{N(N+1)}\left[\sum_{i,j} \ket{i}\bra{j} \otimes B +  \identity_N \otimes C   \right],
\end{equation}
with
\begin{align}
     B=&~\left[\left( N-3\right) \ket{\chi}\bra{\chi} + \ket{\chi}\bra{\chi'} + \ket{\chi'}\bra{\chi} \right], \\C=&~(\ket{\chi'} - \ket{\chi})(\bra{\chi'} - \bra{\chi}).
\end{align}
As the internal states $\ket{\chi}$ and $\ket{\chi'}$ span at most a two-dimensional subspace of $\mathbb{C}^d$, they can be coordinated as
\begin{align}
    \ket{\chi} = ( 1,~0)^T &~~~~~~ \ket{\chi'} = ( e^{i \phi} \cos{\theta},~ e^{i \psi} \sin{\theta})^T.
\end{align}
$M_s$ is now of the form
\begin{equation}
\vspace{-12pt}
    M_s = \frac{1}{N(N+1)}\begin{pmatrix}
     A & B & B &\hdots\\
     B & A & B & \\
     B & B & A & \\
     \vdots&  & & \ddots
    \end{pmatrix},
    \end{equation}
with
\begin{align}
    A =& \begin{pmatrix}
        \cos^2{\theta}+N-2& \sin{\theta}\cos{\theta}~ e^{i(\phi - \psi)} \\ \sin{\theta}\cos{\theta} ~e^{-i(\phi - \psi)} & \sin^2{\theta}
    \end{pmatrix}, \\
    B =& \begin{pmatrix}
        2\cos{\theta}\cos{\phi}+N-3& \sin{\theta}~ e^{- i\psi} \\ \sin{\theta} ~e^{i\psi } & 0
    \end{pmatrix}.
\end{align}

Recall that our goal is the computation of the trace norm of $M_s$, which is given by $M_s=\sum_j |\lambda_j|$, where $\lambda_j$ are the eigenvalues of $M_s$. Since $M_s$ is a ``block circulant'' matrix, its eigenvalues can be obtained by solving \cite{tee_eigenvectors_2007}
\begin{align*}
    &\text{det} (A + (N-1) B - N(N+1)\lambda \identity)\\&\cdot \text{det} (A - B - N(N+1)\lambda \identity)^{N-1} =0. 
\end{align*}
The eigenvalues are accordingly---up to the prefactor of $\frac{1}{N(N+1)}$---given by 0 and \\$2-2 \cos (\theta ) \cos (\phi )$ with $(N-1)$-fold multiplicity each, as well as 
\begin{align*}
\lambda_\pm&=
 \frac{N-1}{2} (2 \cos (\theta ) \cos (\phi )+N-2)\\ \pm&\frac{1}{2}\sqrt{(N-1)^2 (2 \cos (\theta ) \cos (\phi )+N-2)^2+4 N \sin ^2(\theta )},
\end{align*} 
of which one branch is positive and one branch negative. Of particular interest will be the only negative eigenvalue $\lambda_-$ with its eigenvector $\ket{\lambda_-}$, which is generally of the form \begin{equation}
    \ket{\lambda_-} = \frac{1}{\sqrt{N}}(1,~ \omega,~ \omega^2,..., \omega^{N-1}) \otimes \ket{\phi_-}
\end{equation}
with unit $\ket{\phi_-}$ and $\omega$ being a complex root of unity.\\
The trace norm of $M_s$ thus amounts to \begin{widetext}
\begin{equation}
 ||M_s||_1 =\frac{1}{N(N+1)}\bigg[(N-1) (2-2 \cos (\phi) \cos (\theta ))+\sqrt{4 N \sin ^2(\theta )+(N-1)^2 (2 \cos (\phi) \cos (\theta )+N-2)^2}\bigg].
\end{equation}
\end{widetext}
This expression is maximized for 
\begin{align}
    \phi = 0,~~~&\theta = \pi ~~\text{for}~N \leq 3,\nonumber\\ &\theta=\arccos\left(\frac{1-N}{N^2-3 N+1}\right)~\text{else},
\end{align}
which implies that the optimal violation, under the assumption of symmetric unitaries and input state, is given by
\begin{equation}
    \delta_N^{(d)}= \frac{1}{(N+1) \left(N^2-3 N+1\right)} ~~ \text{for}~ N > 3
\end{equation}\changes{as well as $\delta^{(d)}_2 = 1/3$ and $\delta^{(d)}_3 = 1/6$}.



We now also consider the case of a non-uniform input state, i.e. $p_k \neq 1/{N}$ where the win probability depends on the trace norm of $M$. We aim to show that $||M||_1 \leq  ||M_s||_1$. First, consider the case where all $p_k>0$: 
Note that by Sylvester's law of inertia \cite{sylvester_xix_1852}, $M$ also has a single negative eigenvalue $\lambda_{-}'$. We have
\begin{align}
    ||M_s||_1 -  ||M||_1 = \text{Tr}~ M_s + 2 |\lambda_-| -\text{Tr}~M - 2 |\lambda_-'|& \nonumber \\ =2 (|\lambda_-| - |\lambda_-'|) + N \text{Tr}~ A - N \text{Tr}~ A (\sum_{i=1}^N p_i) &\geq 0 ~~ \nonumber\\\Longleftrightarrow~~ |\lambda_-| \geq |\lambda_-'|~~ \Longleftrightarrow~~ M + |\lambda_-|\identity ~&\geq 0 \nonumber \\
    \Longleftrightarrow \frac{N M_s}{|\lambda_-|} +  D_p^{-1} &\geq 0. \label{asymcond2}
\end{align} 
In the last step, we multiplied with $\sqrt{D_p /|\lambda_-|}$ from both sides and used the law of inertia again. Let $\ket{v} = \sum_{i=1}^N \ket{i}\ket{v_i}$ be a general nonzero vector. By using Lagrange multipliers to optimize over $p_i$, it can be shown that 
\begin{equation}
    \bra{v} D_p^{-1} \ket{v} = \sum_{i=1}^N \frac{\braket{v_i }{v_i}}{p_i} \geq \left( \sum_{i=1}^N \sqrt{\braket{v_i}{v_i}}\right)^2.
\end{equation}
Since $M_s$ has only one negative eigenvalue $\lambda_-$, we have that condition \eqref{asymcond2} holds:
\begin{align}
    &\frac{N}{|\lambda_-|} \bra{v} M_s \ket{v} +  \bra{v} D_p^{-1} \ket{v} \nonumber \\&\geq \left( \sum_{i=1}^N\sqrt{\braket{v_i}{v_i}}\right)^2 - N \left \vert \braket{v}{\lambda_-}\right\vert^2 \nonumber\\&\geq 
     \left( \sum_{i=1}^N\sqrt{\braket{v_i}{v_i}}\right)^2 - \frac{N}{K} \left \vert \sum_{i=1}^N \omega^{i-1} \braket{v_i}{\phi_-}\right\vert^2 \nonumber \\
    &\geq \left( \sum_{i=1}^N\sqrt{\braket{v_i}{v_i}}\right)^2 - \frac{N}{K} \left( \sum_{i=1}^N\sqrt{\braket{v_i}{v_i} \braket{\phi_-}{\phi_-}}\right)^2 \geq 0.
\end{align}
To finalize this argument, consider the case where only $m < N$ of the $p_i > 0$. The matrix $M$ can then be modeled as subsystem with a matrix $M^{(m)}$ of $m$ devices padded with zeros:
\begin{equation}
    || M ||_1  = \frac{m +1}{N+1} ||M^{(m)}||_1 < || M^{(m)}_s ||_1 < || M_s ||_1 .
\end{equation}
We have thus shown that with symmetric unitaries, any deviation from a symmetric input state cannot increase the win probability.

\section{Proof of Theorem \ref{theorem:2nd_order_viol}}
\label{appendix_I2}
Let $x_k = P(1| [k])$ denote the second-order interference behavior maximizing the violation $\delta_{I2} = \vecsp{x}{w} - \frac{N}{N+1}$ of the fingerprinting inequality. Now introduce the following change of basis: $\Vec{y} = H \Vec{x}$ with $H$ being the Sylvester-ordered normalized Hadamard matrix. The condition of $\Vec{x}$ being of second order interference now implies that $y_i = 0$ for $\hamwgt{i}> 2$. Since $H$ is orthogonal, we have that $\vecsp{x}{w} = H\vec{x}\cdot H \vec{w} = \vec{y}\cdot{H\vec{w}}$.
Due to permutation symmetry of the fingerprinting inequality and the convexity of $\mathcal{J}_{N,2}$, we can construct a symmetric behavior, 
\begin{equation}
\vec{x}' = \frac{1}{|S_N|}\sum_{R\in S_N} R \vec{x}\in \mathcal{J}_{N,2},
\end{equation} which achieves the same violation as $\vecsp{x}{w} = \vec{x}'\cdot\vec{w}$ holds. Thus, the behavior $\vec{x}$ can without loss of generality be assumed to be permutation-symmetric, implying that its components only depend on the number of ones in the binary representation of $j$: 
\begin{equation}
     x_j = x_k ~~~~~\text{for}~ \hamwgt{j} = \hamwgt{k}.
\end{equation}
But this constraint on $\Vec{x}$ now implies that also $y_i$ depends only on the hamming weight of $i$, meaning there are now only three independent variables to perform optimization over, namely 
\begin{equation}
    z_{\hamwgt{i}} =2^{-N/2} \times y_i,
\end{equation} with $z_i = 0 $, for $i > 2 $.
By directly evaluating $H \Vec{w}$ for the fingerprinting inequality \eqref{fingerprint_ineq}, we obtain
\begin{align}
    (H\Vec{w})_k &= 2^{-N/2} \sum_{n=0}^{2^N-1}  (-1)^{\sum_j [k]_j [n]_j} w_n &\nonumber\\&= 2^{-N/2} \frac{1}{N+1}\left( \sum_{\substack{n=0\\ \hamwgt{n} = 1}}^{2^N-1} (-1)^{\sum_j [k]_j [n]_j} -1 \right) \nonumber\\
   &= \frac{N-1}{N+1} 2^{-N/2}
   ~~~~~~~~~~~\text{for}~\hamwgt{k}=0\nonumber\\
   &= \frac{N-3}{N+1} 2^{-N/2}
   ~~~~~~~~~~~\text{for}~\hamwgt{k}=1\nonumber\\
   &=\frac{N-5}{N+1}  2^{-N/2}~~~~~~~~~~~\text{for}~\hamwgt{k}=2.
\end{align}
\begin{widetext}
Thus taking the scalar product of the transformed quantities, the maximum violation yields
\begin{align}
\delta_{I2} =\frac{1}{N+1} \left( (N-1) z_0 +(N-3) N z_1 +(N-5)\frac{N(N-1)}{2}z_2 - N + 1 \right).
\label{goal_I2_linprog}
\end{align}
The positivity constraints on $\Vec{x}$ also transform as follows:

\begin{align}
 0 \leq x_k &\leq 1 ~~\Longrightarrow  0 \leq (H \Vec{y})_k \leq 1, \text{~with} \\
     x_k& = 2^{-N/2}\sum_{n=0}^{2^N -1} y_n(-1)^{\sum_j [k]_j [n]_j} 
    = z_0 + z_1 \sum_{\substack{n=0\\ \hamwgt{n} = 1}}^{2^N -1}(-1)^{\sum_j [k]_j [n]_j} + z_2 \sum_{\substack{n=0\\ \hamwgt{n} = 2}}^{2^N -1}n(-1)^{\sum_j [k]_j [n]_j} \nonumber\\
    &= z_0 + z_1 \sum_{r=0}^{N-1} (-1)^{\sum_j \delta_{r,j}[k]_j} + z_2  \sum_{\substack{r,s=0 \\ r\neq s}}^{N-1} (-1)^{\sum_j (\delta_{r,j} + \delta_{r,j})[k]_j} \nonumber\\
    &= z_0 + z_1 (N- 2 \hamwgt{k}) + z_2\left( \frac{N(N-1)}{2} + 2\left(N-\hamwgt{k}\right)\hamwgt{k}  \right) ~~~~\in [0, 1]. \label{z_constraint} 
\end{align}
\end{widetext}
In the last step, the sum over $r$ and $s$ amounts to counting the numbers of zeros and ones in the binary expressions $[k]_r$ and $[k]_r \oplus [k]_s$.
The constraint of expression \eqref{z_constraint} being in the interval $[0, 1]$ for all $k$ defines the three-dimensional convex polytope $\polyz$ in the variables $z_i$. Since its bounds depend only on the hamming weight of $k$, the maximum violation $\delta_{I2}$ can be efficiently solved for numerically as a linear program maximizing Eq. \eqref{goal_I2_linprog} over given constraints for each $N$ separately.

We now claim that this linear program of maximizing $\delta_{I2}$ is solved by 
\begin{equation}
 \optz = \frac{\left( \frac{3 \; N^{2} - 7 \; N}{2}, ~~~ N - 3,~~~ -1  \right)^T}{{2 \; N^{2} - 6 \; N + 4}}.
\end{equation}
To this end, we first show that $ \optz$ is indeed a vertex of $\polyz$, and then that it maximizes $\delta_{I2}$. We insert $ \optz$ into the expression \eqref{z_constraint}, obtaining
\begin{equation}
x_k
     =\frac{(N -\hamwgt{k}) (\hamwgt{k}+N-3)}{(N-2) (N-1)} , \label{zstar_inserted}
\end{equation}
which is $\geq 0$ and $\leq 1$  for $k = 0, ... , 2^N -1$,
meaning $\optz\in \polyz$. Furthermore, expression \eqref{zstar_inserted} equals 1 for $\hamwgt{k} = 1, 2$ and equals 0 for $\hamwgt{k} = N, 3-N$. For $N>3$, $\optz$ is the intersection of three facets of $\polyz$ supported by planes, which we denote by $A_1^1, A_2^1$ and $A_N^0$ respectively, making $\optz$
indeed a vertex of $\polyz$.

We now show that $\optz$ maximizes $\delta_{I2}$: If a solution to a convex optimization problem is locally optimal, it is also globally so \cite[138]{boyd_convex_2023}. Thus, due to linearity, a vertex $v$ of a convex polytope maximizes a certain linear functional if and only if all of its neighbor vertices attain a lower value.

Since $\optz$ is the intersection of three facets, it has three neighbor vertices. These lie at the intersection of two faces each. By taking the cross product of the plane's normal vectors, we obtain the following lines of intersection:
\begin{subequations}
\begin{align}
    e_1 &= A_1^1 \cap A_2^1 :\nonumber\\& \optz + t_1 \left(-N^2+5 N-8,2 (N-3),-2\right), \\
    e_2 &= A_1^1 \cap A_N^0 :\nonumber\\& \optz + t_2 \left(N \left(N^2-4 N+3\right),2-2 N,2-2 N\right), \\
    e_3 &= A_2^1 \cap A_N^0 :\nonumber\\& \optz + t_3 \left(-N \left(N^2-7 N+10\right),-8+4 N,-4+2 N\right).
\end{align}
\end{subequations}
We now calculate the change of violation for points on these lines:
\begin{subequations}
\label{deltadeltas}
\begin{align}
\delta_{I2}[e_1(t_1)-\optz] &=
t_1\frac{8}{N+1} & > 0,~ \text{for}~ t_1 > 0, \label{deltatelta1}\\
\delta_{I2}[e_2(t_2)-\optz] &=t_2\frac{4 (N-1) N}{N+1}
 & > 0,~ \text{for}~ t_2 > 0,\\
\delta_{I2}[e_3(t_3)-\optz] &=t_3\frac{4 N \left(N^2-5 N+6\right)}{N+1}  & > 0,~ \text{for}~ t_3 > 0. \label{deltatelta3}
\end{align}
\end{subequations}
Thus, a greater violation than $\optz$ is only possible for $t_i > 0$. By inserting $e_1, e_2, e_3$ into the halfspace inequalities of $A_N^0, A_2^1, A_1^1$, we obtain necessary conditions for $t_i$ in order to describe a point in $\polyz$, which also have to apply to the neighbor vertices of $\optz$: 
\begin{subequations}
\begin{align}
    0 - t_1\cdot (4 N^2-12 N+8) \geq 0 \Longrightarrow t_1 \leq 0,  \\
    1 + t_2\cdot (4 N^2-12 N+8) \leq 1 \Longrightarrow t_2 \leq 0, \\
    1 + t_3\cdot (4 N^2-12 N+8) \leq 1 \Longrightarrow t_3 \leq 0. 
\end{align} 
\end{subequations}
Now note that these conditions together with Eq. \eqref{deltadeltas} imply that $\optz$ has a higher violation $\delta_{I2}$ than all of its neighbor vertices, proving the overall optimality of $\optz$. Inserting $\optz$ into Eq. \eqref{goal_I2_linprog} gives the maximum violation $|\delta_{I2}|$ and Eq. \eqref{zstar_inserted} equals $P(1|[k])$.

\bibliography{content}

\end{document}

%% file: prefix_aps.tex
\usepackage{subcaption}
\usepackage{tikz}
\usepackage[amsmath,thmmarks]{ntheorem}%

\theoremseparator{.}%

\theoremstyle{plain}%
\newtheorem{theorem}{Theorem}[section]

\newtheorem{lemma}[theorem]{Lemma}
\newtheorem{conjecture}[theorem]{Conjecture}

\theoremstyle{plain}%
\theoremheaderfont{\sf} \theorembodyfont{\upshape}%
\newtheorem*{remark:unnumbered}[theorem]{Remark}%
\newcommand{\myqedsymbol}{\rule{2mm}{2mm}}

\theoremheaderfont{\em}%
\theorembodyfont{\upshape}%
\theoremstyle{nonumberplain}%
\theoremseparator{}%
\theoremsymbol{\myqedsymbol}%
\newtheorem{proof}{Proof:}%

\newcommand{\remove}[1]{}%

\newcommand\identity {\openone}
\def\CNK{\mathcal{C}_{N,K}}
\def\Qd{\mathcal{Q}_{N,1}^d}
\def\ket#1{\left\vert #1 \right\rangle}
\def\bra#1{\left\langle #1 \right\vert}
\def\braket#1#2{\left\langle #1 \vert #2\right\rangle}
\newcommand{\hamwgt}[1]{\operatorname{wt} (#1)}
\def\vecsp#1#2{\Vec{#1}\cdot\Vec{#2}}

\newcommand{\lessgtreq}{$\sqcup$\kern-0.58em{$\star$}}

\newcommand{\optz}{\Vec{z}^{\hspace{1.5pt}*}}
\newcommand{\polyz}{\hat{\mathcal{J}}}
\DeclareMathOperator{\Tr}{Tr}
\def\changes#1{ #1}

%% file: content.bib
@article{bibak_quantum_2024,
  title = {Quantum Coherence in Networks},
  author = {Bibak, Fatemeh and Del Santo, Flavio and Daki{\'c}, Borivoje},
  year = {2024},
  month = dec,
  journal = {Physical Review Letters},
  volume = {133},
  number = {23},
  pages = {230201},
  issn = {0031-9007, 1079-7114},
  doi = {10.1103/PhysRevLett.133.230201},
  urldate = {2025-06-07},
  langid = {english}
}

@book{boyd_convex_2023,
  title = {Convex Optimization},
  author = {Boyd, Stephen P. and Vandenberghe, Lieven},
  year = {2023},
  edition = {Version 29},
  publisher = {{Cambridge University Press}},
  address = {{Cambridge New York Melbourne New Delhi Singapore}},
  isbn = {978-0-521-83378-3},
  langid = {english}
}

@article{brunner_bell_2014,
  title = {Bell Nonlocality},
  author = {Brunner, Nicolas and Cavalcanti, Daniel and Pironio, Stefano and Scarani, Valerio and Wehner, Stephanie},
  year = {2014},
  month = apr,
  journal = {Reviews of Modern Physics},
  volume = {86},
  number = {2},
  pages = {419--478},
  issn = {0034-6861, 1539-0756},
  doi = {10.1103/RevModPhys.86.419},
  urldate = {2025-01-16},
  copyright = {http://link.aps.org/licenses/aps-default-license},
  langid = {english}
}

@article{chen_information_2024,
  title = {Information Carried by a Single Particle in Quantum Multiple-Access Channels},
  author = {Chen, Xinan and Zhang, Yujie and Winter, Andreas and Lorenz, Virginia O. and Chitambar, Eric},
  year = {2024},
  month = jun,
  journal = {Physical Review A},
  volume = {109},
  number = {6},
  pages = {062420},
  issn = {2469-9926, 2469-9934},
  doi = {10.1103/PhysRevA.109.062420},
  urldate = {2024-10-07},
  langid = {english}
}

@article{chiribella_quantum_2019,
  title = {Quantum Shannon Theory with Superpositions of Trajectories},
  author = {Chiribella, Giulio and Kristj{\'a}nsson, Hl{\'e}r},
  year = {2019},
  month = may,
  journal = {Proceedings of the Royal Society A: Mathematical, Physical and Engineering Sciences},
  volume = {475},
  number = {2225},
  pages = {20180903},
  issn = {1364-5021, 1471-2946},
  doi = {10.1098/rspa.2018.0903},
  urldate = {2025-06-07},
  langid = {english}
}

@misc{communityproject_polytope_2025,
  title = {Polytope Wiki},
  author = {{community project}},
  year = {2025},
  month = jun,
  urldate = {2025-06-10},
  howpublished = {https://polytope.miraheze.org/wiki/Main\_Page}
}

@article{delsanto_coherence_2020,
  title = {Coherence Equality and Communication in a Quantum Superposition},
  author = {Del Santo, Flavio and Daki{\'c}, Borivoje},
  year = {2020},
  month = may,
  journal = {Physical Review Letters},
  volume = {124},
  number = {19},
  pages = {190501},
  issn = {0031-9007, 1079-7114},
  doi = {10.1103/PhysRevLett.124.190501},
  urldate = {2025-06-07},
  langid = {english}
}

@article{delsanto_twoway_2018,
  title = {Two-Way Communication with a Single Quantum Particle},
  author = {Del Santo, Flavio and Daki{\'c}, Borivoje},
  year = {2018},
  month = feb,
  journal = {Physical Review Letters},
  volume = {120},
  number = {6},
  pages = {060503},
  issn = {0031-9007, 1079-7114},
  doi = {10.1103/PhysRevLett.120.060503},
  urldate = {2025-01-17},
  langid = {english}
}

@article{graczyk_hyperoctahedral_2005,
  title = {The Hyperoctahedral Group, Symmetric Group Representations and the Moments of the Real Wishart Distribution},
  author = {Graczyk, P. and Letac, G. and Massam, H.},
  year = {2005},
  month = jan,
  journal = {Journal of Theoretical Probability},
  volume = {18},
  number = {1},
  pages = {1--42},
  issn = {0894-9840, 1572-9230},
  doi = {10.1007/s10959-004-0579-9},
  urldate = {2024-12-15},
  copyright = {http://www.springer.com/tdm},
  langid = {english}
}

@book{harrison_introduction_1965,
  title = {Introduction to Switching and Automata Theory},
  author = {Harrison, Michael A.},
  year = {1965},
  series = {{{McGraw-Hill}} Series in Systems Science},
  publisher = {{McGraw-Hill}},
  lccn = {64066266}
}

@article{helstrom_quantum_1969,
  title = {Quantum Detection and Estimation Theory},
  author = {Helstrom, Carl W.},
  year = {1969},
  journal = {Journal of Statistical Physics},
  volume = {1},
  number = {2},
  pages = {231--252},
  issn = {0022-4715, 1572-9613},
  doi = {10.1007/BF01007479},
  urldate = {2024-10-07},
  copyright = {http://www.springer.com/tdm},
  langid = {english}
}

@book{horadam_hadamard_2007,
  title = {Hadamard Matrices and Their Applications},
  author = {Horadam, K. J.},
  year = {2007},
  publisher = {{Princeton University Press}},
  address = {{Princeton, N.J}},
  isbn = {978-0-691-11921-2},
  lccn = {QA166.4 .H67 2007}
}

@article{horvat_interference_2021,
  title = {Interference as an Information-Theoretic Game},
  author = {Horvat, Sebastian and Daki{\'c}, Borivoje},
  year = {2021},
  month = mar,
  journal = {Quantum},
  volume = {5},
  pages = {404},
  issn = {2521-327X},
  doi = {10.22331/q-2021-03-08-404},
  urldate = {2025-06-12},
  langid = {english}
}

@mastersthesis{horvat_quantum_2019,
  title = {Quantum Superposition as a Resource for Quantum Communication},
  author = {Horvat, Sebastian},
  year = {2019},
  month = dec,
  address = {{Zagreb}},
  school = {University of Zagreb}
}

@article{horvat_quantum_2021,
  title = {Quantum Enhancement to Information Acquisition Speed},
  author = {Horvat, Sebastian and Daki{\'c}, Borivoje},
  year = {2021},
  month = mar,
  journal = {New Journal of Physics},
  volume = {23},
  number = {3},
  pages = {033008},
  issn = {1367-2630},
  doi = {10.1088/1367-2630/abe9d4},
  urldate = {2023-07-12}
}

@misc{klitzing_polytopes_2025,
  title = {Polytopes \& Their Incidence Matrices},
  author = {Klitzing, Richard},
  year = {2025},
  month = jun,
  urldate = {2025-06-10},
  howpublished = {https://bendwavy.org/klitzing/home.htm}
}

@article{kun_direct_2025,
  title = {Direct and Efficient Detection of Quantum Superposition},
  author = {Kun, Daniel and Str{\"o}mberg, Teodor and Spagnolo, Michele and Daki{\'c}, Borivoje and Rozema, Lee A. and Walther, Philip},
  year = {2025},
  month = may,
  journal = {Physical Review A},
  volume = {111},
  number = {5},
  pages = {L050402},
  issn = {2469-9926, 2469-9934},
  doi = {10.1103/PhysRevA.111.L050402},
  urldate = {2025-06-07},
  langid = {english}
}

@article{lee_higherorder_2017,
  title = {Higher-Order Interference in Extensions of Quantum Theory},
  author = {Lee, Ciar{\'a}n M. and Selby, John H.},
  year = {2017},
  month = jan,
  journal = {Foundations of Physics},
  volume = {47},
  number = {1},
  pages = {89--112},
  issn = {0015-9018, 1572-9516},
  doi = {10.1007/s10701-016-0045-4},
  urldate = {2025-06-07},
  langid = {english}
}

@article{liu_quantum_2023,
  title = {Quantum Communication through Devices with Indefinite Input-Output Direction},
  author = {Liu, Zixuan and Yang, Ming and Chiribella, Giulio},
  year = {2023},
  month = apr,
  journal = {New Journal of Physics},
  volume = {25},
  number = {4},
  pages = {043017},
  issn = {1367-2630},
  doi = {10.1088/1367-2630/acc8f2},
  urldate = {2025-06-07}
}

@article{massa_experimental_2019,
  title = {Experimental Two-way Communication with One Photon},
  author = {Massa, Francesco and Moqanaki, Amir and Baumeler, {\"A}min and Del Santo, Flavio and Kettlewell, Joshua A. and Daki{\'c}, Borivoje and Walther, Philip},
  year = {2019},
  month = nov,
  journal = {Advanced Quantum Technologies},
  volume = {2},
  number = {11},
  pages = {1900050},
  issn = {2511-9044, 2511-9044},
  doi = {10.1002/qute.201900050},
  urldate = {2025-06-07},
  langid = {english}
}

@article{massa_experimental_2022,
  title = {Experimental Semi-Quantum Key Distribution with Classical Users},
  author = {Massa, Francesco and Yadav, Preeti and Moqanaki, Amir and Krawec, Walter O. and Mateus, Paulo and Paunkovi{\'c}, Nikola and Souto, Andr{\'e} and Walther, Philip},
  year = {2022},
  month = sep,
  journal = {Quantum},
  volume = {6},
  pages = {819},
  issn = {2521-327X},
  doi = {10.22331/q-2022-09-22-819},
  urldate = {2025-06-07},
  langid = {english}
}

@book{odonnell_analysis_2014,
  title = {Analysis of Boolean Functions},
  author = {O'Donnell, Ryan},
  year = {2014},
  month = jun,
  edition = {1},
  publisher = {{Cambridge University Press}},
  doi = {10.1017/CBO9781139814782},
  urldate = {2024-02-16},
  isbn = {978-1-107-03832-5 978-1-139-81478-2 978-1-107-47154-2}
}

@article{pawlowski_information_2009,
  title = {Information Causality as a Physical Principle},
  author = {Paw{\l}owski, Marcin and Paterek, Tomasz and Kaszlikowski, Dagomir and Scarani, Valerio and Winter, Andreas and {\.Z}ukowski, Marek},
  year = {2009},
  month = oct,
  journal = {Nature},
  volume = {461},
  number = {7267},
  pages = {1101--1104},
  issn = {0028-0836, 1476-4687},
  doi = {10.1038/nature08400},
  urldate = {2025-01-17},
  copyright = {http://www.springer.com/tdm},
  langid = {english}
}

@article{silva_coherencewitnessing_2023,
  title = {A Coherence-Witnessing Game and Applications to Semi-Device-Independent Quantum Key Distribution},
  author = {Silva, M{\'a}rio and Faleiro, Ricardo and Mateus, Paulo and Cruzeiro, Emmanuel Zambrini},
  year = {2023},
  month = aug,
  journal = {Quantum},
  volume = {7},
  pages = {1090},
  issn = {2521-327X},
  doi = {10.22331/q-2023-08-22-1090},
  urldate = {2025-06-07},
  langid = {english}
}

@article{sorkin_quantum_1994,
  title = {Quantum Mechanics as Quantum Measure Theory},
  author = {Sorkin, Rafael D.},
  year = {1994},
  month = oct,
  journal = {Modern Physics Letters A},
  volume = {09},
  number = {33},
  pages = {3119--3127},
  issn = {0217-7323, 1793-6632},
  doi = {10.1142/S021773239400294X},
  urldate = {2025-06-07},
  langid = {english}
}

@article{sylvester_xix_1852,
  title = {{{XIX}}. {\emph{a Demonstration of the Theorem That Every Homogeneous Quadratic Polynomial Is Reducible by Real Orthogonal Substitutions to the Form of a Sum of Positive and Negative Squares}}},
  author = {Sylvester, J.J.},
  year = {1852},
  month = aug,
  journal = {The London, Edinburgh, and Dublin Philosophical Magazine and Journal of Science},
  volume = {4},
  number = {23},
  pages = {138--142},
  issn = {1941-5982, 1941-5990},
  doi = {10.1080/14786445208647087},
  urldate = {2024-12-15},
  langid = {english}
}

@article{tee_eigenvectors_2007,
  title = {Eigenvectors of Block Circulant and Alternating Circulant Matrices},
  author = {Tee, Garry},
  year = {2007},
  journal = {New Zealand Journal of Mathematics},
  volume = {36},
  pages = {195--211}
}

@article{zhang_building_2022,
  title = {Building Multiple Access Channels with a Single Particle},
  author = {Zhang, Yujie and Chen, Xinan and Chitambar, Eric},
  year = {2022},
  month = feb,
  journal = {Quantum},
  volume = {6},
  eprint = {2006.12475},
  primaryclass = {math-ph, physics:quant-ph},
  pages = {653},
  issn = {2521-327X},
  doi = {10.22331/q-2022-02-16-653},
  urldate = {2023-07-12},
  archiveprefix = {arxiv}
}
